\documentclass[12pt,a4paper]{article}

\usepackage[utf8]{inputenc}
\usepackage[T1]{fontenc}
\usepackage{lmodern}
\usepackage{amsmath,amssymb,amsthm}
\usepackage{graphicx}
\usepackage{caption}
\usepackage{subcaption}
\usepackage{booktabs}
\usepackage{tabularx}

\usepackage{hyperref}
\hypersetup{
	colorlinks=true,
	linkcolor=blue,
	citecolor=blue,
	urlcolor=blue,
}

\usepackage{geometry}
\usepackage{fancyhdr}
\newtheorem{theorem}{Theorem}
\newtheorem{lemma}{Lemma}

\newtheorem{definition}{Definition}

\newtheorem{example}{Example}
\newtheorem{proposition}{Proposition}

\title{%
	\vspace{2cm}
	Packaged Quantum States in Quantum Field Theory:\\
	No Partial Factorization and Gauge-Invariant Multi-Particle Entanglement
}

\author{
	Rongchao Ma \\
	\textit{Department of Physics, University of Alberta, Edmonton, Canada}\\
}

\date{\today}

\begin{document}

	\maketitle

	\begin{abstract}
		We demonstrate that quantum field excitations can generate packaged entangled states, in which all internal quantum numbers (IQNs) (e.g., electric charge, flavor, and color) are inseparably entangled and constrained to irreducible representation (irrep) blocks.
		This is a consequence of local gauge invariance and superselection rules.
		The confinement restricts the net gauge charge to a single superselection sector, thereby excluding cross‐sector superpositions but allowing entanglement within one sector.
		We establish theorems that:
		\textbf{(1)} Explain how these packaged entangled states naturally arise from quantum field excitations,
		\textbf{(2)} Show how they remain gauge invariant or transform covariantly within a fixed net‐charge sector, and
		\textbf{(3)} Illustrate how external degrees of freedom (DOFs) (e.g., spin or momentum) can combine with packaged internal charges to yield gauge-invariant entanglement.
		Finally, we show that spin or momentum measurements on these hybrid states induce a collapse of the internal entanglement.
	\end{abstract}

	\tableofcontents

	\section{Introduction}
	
	\emph{Packaged entangled states} \cite{Ma2017} refer to multi-particle configurations in which all relevant \textbf{internal quantum numbers (IQNs)} are inseparably entangled under the charge-conjugation operator $\hat{C}$. 
	This unique entanglement structure influences how particles are created and annihilated, with potential applications in quantum information.  
	Here, we show that these packaged entangled states can be rigorously derived within quantum field theory, rather than introduced ad hoc.
	
	The striking feature of packaged entanglement is its emergence from two fundamental principles:
	\textbf{(1)} local gauge invariance \cite{Feynman1949,Yang1954,Utiyama1956,Weinberg1967} and 
	\textbf{(2)} superselection rules \cite{WWW1952,DHR1,DHR2,StreaterWightman}. 
	Together, these principles enforce that all IQNs remain locked in \textbf{irreducible representation (irrep)} blocks \cite{Weyl1925,Wigner1939} and prevent partial factorization of internal charges. 
	Moreover, packaged entanglement is not restricted to charge-conjugation alone or to particle-antiparticle pairs.
	It also arises among identical particles, distinct particles, and multi-particle systems of any size.
	
	In this paper, we demonstrate:
	\textbf{(1)} \textit{Single-particle level.} 
	Local gauge invariance compels each creation operator to form an irrep of $\mathrm{gauge}\times\mathrm{Lorentz}$.  
	Consequently, the IQNs of a single-particle remain inseparable and prohibit partial factorization.
	We formalize this in Theorem~\ref{thm:NoPartialFactorization}.
	\textbf{(2)} \textit{Multi-particle level.}
	In multi-particle systems, each creation operator still constitutes its own $\mathrm{gauge} \times \mathrm{Lorentz}$ irrep and superselection rules fix the total charge (or color) within a single superselection sector.
	This constraint enables nontrivial packaged entangled states \cite{Ma2017} yet disallows cross‐sector interference.
	We prove in Theorem~\ref{thm:SuperselectionPackaging} that such multi-particle packaged entangled states remain gauge invariant (or, more precisely, transform covariantly) and confined to one superselection sector.
	\textbf{(3)} \textit{Hybrid states.}
	We further consider states that combine IQNs and external \textbf{degrees of freedom (DOFs)} (spin or momentum). 
	Theorem~\ref{thm:HybridGauge} shows how measuring an external DOF can collapse the internal entanglement, illustrating the interplay between gauge constraints and quantum measurements.

	\section{Single-Particle Packaging}
	
	Consider a field $\psi(x)$ in a local gauge theory \cite{PeskinSchroeder,WeinbergBook} with gauge group $G$, typically $U(1)$ or $SU(N)$.  
	Under canonical quantization, $\psi(x)$ expands in creation/annihilation operators:
	\begin{equation}\label{eq:fieldExpansion}
		\psi(x)
		\;=\;
		\int\!\!\frac{d^3p}{(2\pi)^3}\sum_s 
		\Bigl[
		\hat{a}_s(\mathbf{p})\,u_s(\mathbf{p})\,e^{-ip\cdot x}
		\;+\;
		\hat{b}_s^\dagger(\mathbf{p})\,v_s(\mathbf{p})\,e^{+ip\cdot x}
		\Bigr],
	\end{equation}	
	where $\hat{a}_s^\dagger(\mathbf{p})$ and $\hat{b}_s^\dagger(\mathbf{p})$ create particle/antiparticle excitations with spin projection $s$.  
	Here $s$ typically indexes an external DOF (spin or helicity), while the operators carry internal gauge charges (e.g., electron vs.\ positron, quark vs.\ antiquark, etc.).  
	By design, $\psi$ belongs to an irrep of $G \times \text{Lorentz}$.

	\paragraph{(1) Definition of Single-Particle Packaged States.}
	We now formally give the definition of a single-particle packaged state:
	\begin{definition}[Single-Particle Packaged State]
		\label{def:singlePackaging}
		Consider a single-particle state
		\begin{equation}\label{SingleParticlePackagedState}
			|P\rangle \;\equiv\; \hat{a}^\dagger(\mathbf{p})\,|0\rangle \text{ or }
			\lvert \bar{P}\rangle\equiv \hat{b}^\dagger(\mathbf{p})\,\lvert 0\rangle.
		\end{equation}
		If the operator $\hat{a}^\dagger(\mathbf{p})$ transforms as an irrep under local gauge transformations and thereby carries all relevant IQNs (electric charge, flavor, color, baryon/lepton number, etc.) as one inseparable block,
		then we say that state $|P\rangle$ is a \textbf{single-particle packaged state}.
	\end{definition}
	
	\begin{example}[Examples of Single-Particle Packaged States]
		Typical examples include:
		\begin{enumerate}
			\item The electron field $\psi_e(x)$ in QED carrying $U(1)$ charge $-e$ and spin-$\tfrac12$.  
			The creation operator $\hat{a}_{e^-}^\dagger(\mathbf{p})$ forms a single irrep: spin-$\tfrac12$ \emph{plus} electric charge $-e$.
			
			\item A quark field in QCD carrying color $\mathbf{3}$ or $\overline{\mathbf{3}}$.
			
			\item A lepton field with fixed lepton number.
			
			In all cases, the field operator transforms as an irreducible representation of $\mathrm{gauge}\times\mathrm{Lorentz}$, so one cannot independently manipulate the charge or color component of $\hat{a}^\dagger$.
		\end{enumerate}
	\end{example}

	\paragraph{(2) Packaging of Single-Particle IQNs.}
	Why no partial factorization of IQNs?
	The reason is that local gauge invariance forbids splitting these quantum numbers.  
	We restate this standard fact in an information-theoretic language as follows:
	
	\begin{theorem}[Packaging of Single-Particle IQNs]
		\label{thm:NoPartialFactorization}
		Let $\psi(x)$ be an irreducible field under $G \times \mathrm{Lorentz}$, where $G$ is a local gauge group (Abelian or non-Abelian).
		Then each creation operator $\hat{a}^\dagger(\mathbf{p})$ is a single, complete irrep block that packages \emph{all} IQNs (electric charge, flavor, color, etc.), disallowing partial factorization.
	\end{theorem}
	
	\begin{proof}
		Since the field operators belong to an irreducible representation of $G \times \mathrm{Lorentz}$ (in the sense of the Wightman axioms or an equivalent approach), Schur's lemma applies: one cannot decompose an irreducible representation into smaller invariant sub-representations.  
		Hence, any attempt to split the charge (or color, flavor) among multiple substates within a single-particle would violate the representation structure.  
		Therefore, irreducibility directly implies the ``packaging'' of these IQNs.
	\end{proof}

	\paragraph{(3) Superposition of Single-Particle Packaged States.}
	Single‐particle irreps cannot exhibit entanglement because, by definition, entanglement requires correlations between distinct subsystems (i.e., multiple particles).
	Moreover, if a particle has a nonzero gauge charge, superselection forbids forming a coherent superposition of particle and antiparticle states, for example, \(\alpha\,|P\rangle + \beta\,|\bar{P}\rangle\) would be disallowed if \(P\) and \(\bar{P}\) carry opposite (nonzero) charges.
	However, single‐particle superpositions of the form \(\alpha\,|P\rangle + \beta\,|\bar{P}\rangle\) are allowed in the following circumstances:
	\begin{enumerate}
		\item No net gauge charge. If both \(P\) and \(\bar{P}\) carry zero gauge charge (e.g., a genuinely neutral particle‐antiparticle pair), then they lie in the same superselection sector. No superselection barrier separates \(|P\rangle\) and \(|\bar{P}\rangle\).
		
		\item Difference is global, not gauged. If \(P\) and \(\bar{P}\) differ only by a global quantum number (such as flavor), then that difference is not protected by a local gauge symmetry.
	\end{enumerate}
	
	\begin{example}[Superposition of Neutral Mesons]
		Neutral mesons can indeed form superpositions like 
		\[
		\alpha\,\lvert K^0\rangle \;+\; \beta\,\lvert \overline{K}^0\rangle,
		\]
		where \(\lvert K^0\rangle\) and \(\lvert \overline{K}^0\rangle\) are flavor eigenstates but both carry zero net electric charge.
	\end{example}

	Thus, at the single‐particle level, local gauge symmetry enforces that all internal charges remain ``locked together'' within one creation operator, forming an irreducible representation.
	We refer to this as a ``packaged'' irrep.
	In the next section, we see how this principle naturally extends to multi‐particle superpositions.

	\section{Multi-Particle Packaging}
	
	We now consider the packaging of multi-particle states.
	
	\paragraph{(1) Multi-Particle Product States.}
	The full Fock space $\mathcal{H}_\text{Fock}$ is spanned by applying creation operators to the vacuum:
	\begin{equation}\label{FockSpace}
		\mathcal{H}_\text{Fock}
		\;=\;
		\Bigl\{
		\,\hat{a}^\dagger(\mathbf{p}_1)\cdots \hat{a}^\dagger(\mathbf{p}_n)\,\lvert 0\rangle,\;
		\hat{b}^\dagger(\mathbf{q}_1)\cdots \hat{b}^\dagger(\mathbf{q}_m)\,\lvert 0\rangle,\;\dots
		\Bigr\},
	\end{equation}
	where each $\hat{a}^\dagger(\mathbf{p}_i)$ create a particle with momentum $\mathbf{p}_i$ and $\hat{b}^\dagger(\mathbf{q}_j)$ create an antiparticle with momentum $\mathbf{q}_j$.
	Since both $\hat{a}^\dagger(\mathbf{p}_i)$ and $\hat{b}^\dagger(\mathbf{q}_j)$ serve as creation operators, for simplicity we will denote by $\hat{a}^\dagger(\mathbf{p}_i)$ the creation of either a particle or an antiparticle.	
	
	For later use, we first introduce the following lemma:
	
	\begin{lemma}[Gauge Covariance of Multi-Particle Product States]
		\label{lemma:GaugeInvMultpParticle}
		Let $G$ be a local gauge group (Abelian or non-Abelian).  
		Consider a multi-particle product state
		\[
		|\Theta\rangle
		\;=\;
		\hat{a}_1^\dagger(\mathbf{p}_1,q_1)\,
		\hat{a}_2^\dagger(\mathbf{p}_2,q_2)\,\cdots\,
		\hat{a}_n^\dagger(\mathbf{p}_n,q_n)\,\lvert 0\rangle,
		\]
		where each $\hat{a}_k^\dagger(\mathbf{p}_k,q_k)$ transforms in an irreducible representation of $G$ (i.e., it packages all internal quantum numbers of that particle).  
		Then $|\Theta\rangle$ transforms in a well-defined (generally reducible) representation of $G$, namely the tensor product of the individual single-particle irreps.  
		In particular, if $G=U(1)$, then $|\Theta\rangle$ carries total charge $Q=\sum_k q_k$.
	\end{lemma}
	
	\begin{proof}
		We use standard arguments from quantum field theory:
		\begin{enumerate}
			\item \emph{Single-particle irreps.}  
			By Theorem~\ref{thm:NoPartialFactorization}, each creation operator 
			$\hat{a}_k^\dagger(\mathbf{p}_k,q_k)$ transforms under an irreducible representation of $G$.  
			Concretely, for each gauge transformation $g\in G$,
			\[
			U_g\,\hat{a}_k^\dagger(\mathbf{p}_k,q_k)\,U_g^{-1}
			\;=\;
			\sum_{q'} \bigl[D^{(k)}_g\bigr]_{q'\,q_k}\,
			\hat{a}_k^\dagger(\mathbf{p}_k,q'),
			\]
			where $[D^{(k)}_g]$ is the matrix representation of $g$ in the $k$th single-particle irrep.
			
			\item \emph{Product operator.}  
			Since each $\hat{a}_k^\dagger$ transforms by $D^{(k)}$, the product operator
			\[
			\hat{a}_1^\dagger(\mathbf{p}_1,q_1)\;
			\hat{a}_2^\dagger(\mathbf{p}_2,q_2)\;\cdots\;
			\hat{a}_n^\dagger(\mathbf{p}_n,q_n)
			\]
			transforms according to the tensor product
			\[
			D^{(1)} \;\otimes\; D^{(2)} \;\otimes\;\cdots\;\otimes\; D^{(n)}.
			\]
			Acting on the vacuum $|0\rangle$ (which is gauge invariant), we obtain
			\[
			|\Theta\rangle
			\;=\;
			\hat{a}_1^\dagger\,\cdots\,\hat{a}_n^\dagger \,|0\rangle,
			\]
			which thus lies in a definite representation of $G$ (generally reducible if $G$ is non-Abelian).
			
			\item \emph{Net charge or color.}  
			In an Abelian gauge theory $G=U(1)$, the tensor product is labeled by the sum of charges, $Q=\sum_k q_k$.  
			Hence $|\Theta\rangle$ has total charge $Q$.  
			In a non-Abelian theory (e.g., $SU(3)$ for color), the tensor product of irreps may decompose into a direct sum:
			\[
			D^{(1)} \otimes \cdots \otimes D^{(n)}
			\;=\;
			\bigoplus_\lambda\, N_{\lambda}\, D_{\lambda}.
			\]
			Either way, $|\Theta\rangle$ transforms within a well-defined gauge representation.
		\end{enumerate}
	\end{proof}

	\paragraph{(2) Superselection and Superpositions.}
	Superselection theorems (Wick-Wightman-Wigner \cite{WWW1952}, Doplicher-Haag-Roberts \cite{DHR1,DHR2}, etc.) imply that net-charge sectors $\mathcal{H}_Q$ for $Q\neq Q'$ are disjoint Hilbert subspaces that cannot form coherent superpositions.  
	However, within a single sector $\mathcal{H}_Q$, linear combinations of multi-particle states remain valid and can produce nontrivial packaged entanglement.
	
	Hence, we can form superpositions of these basis states if they all lie within the same net-charge sector, which ensure that superselection is respected.  
	The total state remains in a single net-charge sector $Q$ even if individual particles carry different single-particle charges.  
	If the resulting superposition is non-factorizable across the constituent excitations with respect to internal charges, we obtain a multi-particle packaged entangled state.

	\paragraph{(3) Definition of Multi-Particle Packaged Entangled States.}
	Now we are ready to give a formal definition for multi-particle packaged entangled states:
	
	\begin{definition}[Multi-Particle Packaged Entangled State]
		\label{def:MultiPackaged}
		Consider a multi-particle state
		\begin{equation}\label{SuperpositionState}
			|\Psi\rangle \;=\; \sum_n \alpha_n \,|\Theta_n\rangle,
		\end{equation}
		where each
		$
		|\Theta_n\rangle
		\;=\;
		\hat{a}_{n,1}^\dagger\,\hat{a}_{n,2}^\dagger \cdots \,\lvert 0\rangle
		$
		is a multi-particle basis state.
		If the state $|\Psi\rangle$ satisfies the following conditions:
		\begin{enumerate}
			\item Each $\hat{a}_{n,k}^\dagger$ is a single-particle packaged operator (Definition~\ref{def:singlePackaging});
			
			\item Every basis state $|\Theta_n\rangle$ belongs to the same net-charge subspace $\mathcal{H}_Q$, so superselection is respected;
			
			\item The total wavefunction is non-factorizable across its multiple excitations with respect to internal charges (i.e., internal entangled);
		\end{enumerate}
		then we say that $|\Psi\rangle$ is a \textbf{packaged entangled state}.  
		Otherwise, $|\Psi\rangle$ is a \textbf{packaged non-entangled state}.
	\end{definition}
	
	\begin{example}[Electron-Positron Pair vs. Forbidden $\pm 2e$ Mixing]
		Let $\hat{a}_{e^-}^\dagger$ create an electron ($Q=-e$) and $\hat{b}_{e^+}^\dagger$ create a positron ($Q=+e$).  
		Then
		\[
		\alpha \,\hat{a}_{e^-}^\dagger(\mathbf{p}_1)\,\hat{b}_{e^+}^\dagger(\mathbf{p}_2)\,|0\rangle
		\;+\;
		\beta \,\hat{b}_{e^+}^\dagger(\mathbf{p}_1)\,\hat{a}_{e^-}^\dagger(\mathbf{p}_2)\,|0\rangle
		\]
		lies in $\mathcal{H}_{Q=0}$ and forms a packaged entangled state for $\alpha,\beta\neq 0$.
		By contrast, mixing $\hat{a}_{e^-}^\dagger\hat{a}_{e^-}^\dagger$ ($Q=-2e$) with $\hat{b}_{e^+}^\dagger\hat{b}_{e^+}^\dagger$ ($Q=+2e$) violates superselection because it would superpose net charges $\pm 2e$.
	\end{example}
	
	\paragraph{(4) Packaging of Multi-Particle IQNs.}
	We now establish the main properties of multi-particle packaged states:
	
	\begin{theorem}[Packaging of Multi-Particle IQNs]
		\label{thm:SuperselectionPackaging}
		Let
		$
		|\Theta_n\rangle
		\;=\;
		\hat{a}_{n,1}^\dagger\,\hat{a}_{n,2}^\dagger \cdots \,\lvert 0\rangle
		$
		be multi-particle basis states belonging to the same net-charge subspace $\mathcal{H}_Q$.  
		Consider any superposition
		\[
		|\Psi\rangle 
		\;=\; 
		\sum_n \alpha_n \,|\Theta_n\rangle,
		\quad \forall ~ |\Theta_n\rangle\in\mathcal{H}_Q.
		\]
		Then:
		\begin{enumerate}
			\item \textbf{Superselection Compatibility:}  
			$|\Psi\rangle$ is physically realizable with no cross-sector interference.
			Mixing states from $\mathcal{H}_Q$ and $\mathcal{H}_{Q'}$ (for $Q\neq Q'$) is forbidden.
			
			\item \textbf{Packaged Entangled State:}  
			If $|\Psi\rangle$ is non-factorizable over its particles, it is called a \emph{packaged entangled state}.
			
			\item \textbf{Gauge Covariance:}  
			Since each $|\Theta_n\rangle$ transforms covariantly with total charge $Q$ (Lemma~\ref{lemma:GaugeInvMultpParticle}), any linear combination $|\Psi\rangle$ within the same charge sector also transforms in a well-defined gauge representation.  
			No partial IQNs can be extracted from a single term.
		\end{enumerate}
	\end{theorem}
	
	\begin{proof}
		We proceed in three steps:
		\begin{enumerate}
			\item \emph{Superselection compatibility.}  
			By the Wick-Wightman-Wigner \cite{WWW1952} and Doplicher-Haag-Roberts \cite{DHR1,DHR2} theorems, states carrying different net charges $Q \neq Q'$ are orthogonal under local gauge-invariant measurements.  
			Consequently, if each $|\Theta_n\rangle\in \mathcal{H}_Q$, their linear superposition $\sum_n \alpha_n\,|\Theta_n\rangle$ remains strictly in $\mathcal{H}_Q$ and cannot mix with $\mathcal{H}_{Q'}$.
			
			\item \emph{Packaged entanglement.}  
			A superposition $|\Psi\rangle$ is \emph{entangled} if it cannot be written as a tensor product across distinct excitations, while still lying in one net-charge sector $Q$.  
			In that case, we call $|\Psi\rangle$ a \emph{packaged entangled state} (Definition~\ref{def:MultiPackaged}).
			
			\item \emph{Gauge covariance.}  
			Each basis state $|\Theta_n\rangle = \hat{a}_{n,1}^\dagger \hat{a}_{n,2}^\dagger \cdots\,|0\rangle$ transforms as a definite (possibly reducible) representation with total charge $Q$ (Lemma~\ref{lemma:GaugeInvMultpParticle}).  
			Since $|\Psi\rangle$ is a linear combination of such vectors in the same charge sector, it remains in a well-defined gauge representation with net $Q$.
		\end{enumerate}
		
		Thus, $|\Psi\rangle$ is both superselection-compatible and gauge-covariant.  
		If it is non-separable over its particles, it qualifies as a \emph{packaged entangled state}.
	\end{proof}

	\paragraph{(5) Existence of a Packaged Entangled Basis.}
	With superselection rules, we can decompose the total Hilbert space factors as
	\[
	\mathcal{H}
	\;=\;
	\bigoplus_{Q \,\in\, \sigma(\hat{Q})}\,
	\mathcal{H}_Q,
	\]
	where $\hat{Q}$ is the self-adjoint charge operator and each $\mathcal{H}_Q$ is the subspace of states with net charge $Q$.
	In fact, $\mathcal{H}_Q$ is exactly the set of all multi-particle packaged states with total charge $Q$.
	The states in \(\mathcal{H}_Q\) can form packaged entangled states, denote by $E_Q$.  
	Since the sum of two packaged entangled states can yield a product (separable) state, $E_Q$ is not closed under addition and hence does not form a linear subspace.	
	Then the question is: Can we pick out a complete orthonormal basis from $E_Q$ that spans \(\mathcal{H}_Q\) ?
	The answer is yes.
	
	\begin{proposition}[Maximal Orthonormal Packaged Entangled Basis]
		\label{ExistenceOfMaximalOrthonormalPEB_v2}
		Let \(\mathcal{H}_Q\) be the subspace of total charge \(Q\) in a gauge theory.  Then there exists a complete orthonormal basis of \(\mathcal{H}_Q\) whose elements are all packaged entangled states.
	\end{proposition}
	
	\begin{proof}
		We now prove the proposition using Gram-Schmidt procedure:
		
		\begin{enumerate}
			\item Span \(\mathcal{H}_Q\) by Multi‐Particle Creation Operators.
			Since $\mathcal{H}_Q$ is the set of all packaged states with total charge $Q$, we can span $\mathcal{H}_Q$ using packaged states as follows,
			\begin{equation}\label{eq:SpanHQWithProductStates}
				\Bigl\{
				\hat{a}^\dagger(\mathbf{p}_1,q_1)\,\hat{a}^\dagger(\mathbf{p}_2,q_2)\,\cdots\,|0\rangle
				\;\Big\vert\;
				\sum_i q_i \;=\; Q
				\Bigr\}.
			\end{equation}
			Each such product is a “packaged” multi‐particle state in the sense of local gauge invariance.
			
			\item Constructing Orthonormal Packaged Entangle States.			
			At each step $k$, we use Gram-Schmidt procedure to construct the new packaged entangle state \(\Psi_k\) that is orthogonal to all preceding packaged entangled states \(\Psi_1, \ldots, \Psi_{k-1}\).
			This ensure that \(\Psi_k\) cannot be factorized across any bipartition of the creation operators.
			Such states exist as soon as \(\dim \mathcal{H}_Q > 1\).
			After normalization, \(\Psi_k\) has length 1.
			Thus we create an orthonormal set \(\{ \Psi_1, \Psi_2, \ldots, \Psi_n \}\).
			
			\item Completeness of Packaged Entangle States.
			From Eq.(\ref{eq:SpanHQWithProductStates}), we can always find a packaged state set $\{ \Theta_1, \Theta_2, \ldots, \Theta_n \}$ as the basis of $\mathcal{H}_Q$.
			When using Gram–Schmidt procedure, we do not change subspace $\mathcal{H}_Q$.
			Although each new vector \(\Psi_k\) is a modified version of \(\Theta_k\), it still belongs to the same subspace $\mathcal{H}_Q$.
			Consequently, \(\{ \Psi_1, \Psi_2, \ldots, \Psi_n \}\) still spans $\mathcal{H}_Q$.  
			Furthermore, the construction preserves linear independence: if the \(\Theta_k\) were linearly independent, so are the resulting orthonormal packaged entangled states \(\Psi_k\).
			Thus they form a basis of $\mathcal{H}_Q$.  
			Putting it together, \(\{ \Psi_1, \Psi_2, \ldots, \Psi_n \}\) is a complete orthonormal basis for $\mathcal{H}_Q$.			
		\end{enumerate}
		
		This shows that we can turns any linearly independent set of packaged states in $\mathcal{H}_Q$ into a complete orthonormal basis of packaged entangle states in $\mathcal{H}_Q$.
	\end{proof}
	
	This process is crucial for applications in quantum information because it guarantees the existence of a maximal orthonormal set $\{\lvert \Psi_i\rangle\}$ of packaged entangled states spanning $\mathcal{H}_Q$.  
	Consequently, Bell-like measurements on $\mathcal{H}_Q$ become possible.
	
	\begin{example}[Maximal Orthonormal Basis of Packaged Entangled States]
		\label{exm:MaximalOrthonormalBasisOfPES}
		Consider a particle-antiparticle system (e.g., an electron and a positron).  
		A simple basis for the neutral sector $Q=0$ is given by the two orthonormal states
		\begin{equation}\label{eq:PESParticleAntiparticle}
			\bigl\lvert \Psi^\pm \bigr\rangle 
			\;=\;
			\tfrac{1}{\sqrt2}\,\Bigl[
			\hat{a}^\dagger(\mathbf{p}_1)\,\hat{b}^\dagger(\mathbf{p}_2)
			\;\pm\;
			\hat{b}^\dagger(\mathbf{p}_1)\,\hat{a}^\dagger(\mathbf{p}_2)
			\Bigr]
			\lvert 0\rangle,
		\end{equation}
		where $\hat{a}^\dagger$ and $\hat{b}^\dagger$ create opposite charges $\pm e$.  
		Each $\lvert \Psi^\pm\rangle$ is a packaged entangled state, and $\{\lvert\Psi^+\rangle,\lvert\Psi^-\rangle\}$ forms an orthonormal basis for $\mathcal{H}_{Q=0}$ in this simplified two-particle model.
	\end{example}

	\section{Hybridization of Internal \& External Entanglement}
	
	While internal charges (e.g., electric charge, color) must remain ``packaged'', real particles also carry external DOFs such as spin ($s$) or momentum ($\mathbf{p}$) that are not necessarily gauged.
	If spin does not transform non-trivially under $G$, then it is an independent factor in the single-particle representation (e.g., spin-$\tfrac12$ $\otimes$ charge $-e$).
	Hence a single-particle operator $\hat{a}^\dagger_{s,q}(\mathbf{p})$ can carry both spin $s$ (external) and gauge charge $q$ (internal).  
	Similarly, momentum ($\mathbf{p}$) or other quantum numbers can appear.
	Here we show that multiple excitations can form states that entangle spin and IQNs across different particles without violating gauge invariance.

	\paragraph{(1) Definition of Hybridized Packaged Entangled States.}
	We now give the formally definition of hybridized packaged entangled states:
	
	\begin{definition}[Hybridized Packaged Entangled State]
		\label{def:HybridPackaged}
		Consider a multi-particle state
		\[
		|\Psi\rangle \;=\; \sum_{n}\alpha_n\;
		\Bigl(
		\hat{a}_{n,1}^\dagger(s_1,q_1)\,
		\hat{a}_{n,2}^\dagger(s_2,q_2)\cdots
		\hat{a}_{n,m}^\dagger(s_m,q_m)
		\Bigr)\,|0\rangle,
		\]
		where $n, m$ label distinct particles, $s$ is the spin index (or helicity) of a particle, and $q$ is the particle’s internal gauge charge (e.g., $\pm e$, color $\mathbf{3}$ or $\overline{\mathbf{3}}$).  
		If $|\Psi\rangle$ satisfies the following conditions:
		\begin{enumerate}
			\item Each single-particle operator $\hat{a}^\dagger(s,q)$ is a packaged irrep carrying the full internal charge $q$ and external spin $s$;
			
			\item All terms 
			$\hat{a}_{n,1}^\dagger(s_1,q_1)\,
			\hat{a}_{n,2}^\dagger(s_2,q_2)\cdots
			\hat{a}_{n,m}^\dagger(s_m,q_m)
			\,|0\rangle$
			are multi-particle basis states that lie in the same net-charge sector (e.g., net charge $Q$), so superselection is respected;
			
			\item The total wavefunction is non-factorizable across various excitations with respect to internal charges and external DOFs (spin, momentum) (i.e., it is entangled both internally and externally);
		\end{enumerate}
		then we say that $|\Psi\rangle$ is a \textbf{hybridized packaged entangled state}.
	\end{definition}
	
	\begin{example}[Hybrid Spin-Charge Entangled Pair]
		\label{ex:HybridExample}
		As a concrete illustration, consider an electron-positron pair, where each particle can be spin-up $\uparrow$ or spin-down $\downarrow$.  
		A simple hybridized packaged entangled state is:				
		\[
		\alpha \,\hat{a}_{e^-,\uparrow}^\dagger(\mathbf{p}_1)\,\hat{b}_{e^+,\downarrow}^\dagger(\mathbf{p}_2)\,|0\rangle
		\;+\;
		\beta \,\hat{b}_{e^+,\downarrow}^\dagger(\mathbf{p}_1)\,\hat{a}_{e^-,\uparrow}^\dagger(\mathbf{p}_2)\,|0\rangle
		\]
		with both terms lying in the net $Q=0$ sector for $\alpha,\beta\neq 0$. 
		Each creation operator $\hat{a}_{e^-,\uparrow}^\dagger$ or $\hat{b}_{e^+,\downarrow}^\dagger$ is a packaged operator carrying (charge $-e$, spin $\uparrow$) or (charge $+e$, spin $\downarrow$).  
		The entanglement is hybrid because measuring spin on one particle projects the entire spin-charge wavefunction for both particles.
	\end{example}

	\paragraph{(2) Packaging of Hybrid Internal-External DOFs.}
	We conclude with a statement on how gauge invariance is preserved in hybrid states, yet spin or momentum measurements can collapse internal DOFs if the state is hybridized entangled:
	
	\begin{theorem}[Packaging of Hybrid Internal-External DOFs]
		\label{thm:HybridGauge}
		Consider internal IQNs (electric charge or color) and external DOFs (spin or momentum) appended to each creation operator.
		Then:
		\begin{enumerate}
			\item The total state is physically realizable with no cross-sector interference.  
			Mixing $\mathcal{H}_Q$ with $\mathcal{H}_{Q'}$ (for $Q\neq Q'$) is disallowed.
			
			\item The total state remains gauge invariant (or, more precisely, transforms covariantly within the same gauge sector) as long as the net gauge charge $Q$ of each term is fixed within the superposition.
			
			\item If the total wavefunction is hybridized entangled across internal IQNs and external DOFs, then a projective measurement on the external DOFs can collapse the entire external-internal wavefunction, i.e., it forces both the external DOFs and the internal IQNs into a definite configuration while preserving the net charge.
		\end{enumerate}
	\end{theorem}
	
	\begin{proof}
		We proceed in three steps:
		
		\begin{enumerate}
			\item Since the net gauge charge is unchanged across all superposition terms, superselection prevents mixing between different charge sectors.
			
			\item A local gauge transformation does not affect spin (or momentum) because spin is associated with an external $\mathrm{SU}(2)$ symmetry that is decoupled from the gauge group (e.g., $U(1)$ or $SU(3)$).  
			Thus, as long as the net charge $Q$ remains the same in every term, the overall wavefunction transforms within the same gauge sector.
			
			\item Although a measurement on spin (or momentum) is not a gauge transformation, it is a valid physical observable.  
			Since each creation operator carries both spin and gauge charge, measuring the spin projects the entire wavefunction.  
			Even though superselection guarantees that the total gauge charge remains fixed, the spin measurement collapses any existing spin-gauge entanglement and thereby forces the internal charges into a definite configuration within that net-charge sector.
		\end{enumerate}
		
		Thus, the entire ``external $\otimes$ internal'' wavefunction remains gauge invariant (i.e., confined to a single gauge sector), and a measurement of an external DOF triggers wavefunction collapse across the internal IQNs as well \cite{NielsenChuang}.
	\end{proof}
	
	In Example~\ref{ex:HybridExample}, measuring spin-up versus spin-down on the electron forces the positron’s spin to collapse, thereby collapsing the entire correlated spin-charge state.
	The net $Q=0$ sector is preserved, but the post-measurement outcome can break the entanglement structure.
	This is the typical phenomenon observed in bipartite entanglement, now embedded within a gauge invariant context.
	
	Recent work on lattice gauge theories \cite{Zohar2016,Sala2018} and on quantum resource allocation in gauge-invariant systems has highlighted how gauge constraints restrict accessible states in a fundamental way.  
	Our packaging framework clarifies a central point: partial or fractional IQNs cannot appear as independent quantum DOFs in a single-particle state.  
	It further shows how gauge singlets (e.g., color-singlet quark-antiquark states) can produce packaged entangled pairs within one net-charge sector.  
	Such states exemplify packaged entanglement: the wavefunction is packaged entangled while still respecting superselection and local gauge invariance.

	\section{Packaging Principle}
		
	We now consolidate Theorems~\ref{thm:NoPartialFactorization}, \ref{thm:SuperselectionPackaging}, and~\ref{thm:HybridGauge} into one unifying statement:
	
	\begin{definition}[Packaging Principle]
		\label{def:PackagingPrinciple}
		In a local gauge theory subject to superselection rules, all \textbf{internal quantum numbers (IQNs)} (electric charge, flavor, or color) are packaged into irreducible representation (irrep) blocks under the gauge group.
		Consequently:
		\begin{enumerate}
			\item \textbf{No partial factorization of IQNs:}  
			Single-particle creation operators must transform as complete irreps of the gauge group.  
			One cannot split or partially factor out the IQNs from a single-particle excitation (Theorem~\ref{thm:NoPartialFactorization}).
			
			\item \textbf{Single net-charge superselection sector:}  
			Physical states cannot coherently superpose different net charges.  
			All multi-particle superpositions must lie within a single net-charge sector (Theorem~\ref{thm:SuperselectionPackaging}).
			
			\item \textbf{Packaged entanglement:}  
			Within a fixed charge sector, multi-particle excitations can form entangled superpositions.  
			We call this ``packaged entanglement'' because it spans all the IQNs (Theorem~\ref{thm:SuperselectionPackaging}).
			One cannot entangle part of the IQNs while leaving others separable.
			
			\item \textbf{Hybrid internal-external entanglement:}  
			Since external degrees of freedom (e.g., spin or momentum) do not transform under the local gauge group, they can be non-trivially combined with the IQNs, permitting \emph{hybrid} packaged entangled states.  
			A measurement of external DOFs can collapse the overall wavefunction, including its internal charges, while preserving the original net-charge sector (Theorem~\ref{thm:HybridGauge}).
		\end{enumerate}
	\end{definition}
	
	In short, this packaging principle unifies four features of packaged entanglement we have analyzed.  
	It shows that all IQNs are ``locked together'' in irreps, while external DOFs remain free to entangle, be measured, and collapse the entire wavefunction within a single charge sector.

	\section{Discussion}
	
	We have shown that packaged entangled states are natural consequence of local gauge invariance and superselection rules.
	This bridges standard field-theoretic constraints with modern entanglement measures.
	Packaged entangled states may be used as robust quantum information carriers and in quantum error-correction protocols that inherently respect gauge symmetry.
	The conventional approaches usually treat entanglement independently of gauge considerations.
	Here we highlight that the very structure of a particle’s creation operator and thus the inseparability of its internal quantum numbers can serve as a powerful resource for quantum information processing.	
	We anticipate these results will be useful for:
	(1) \emph{Quantum simulations of lattice gauge theories} \cite{Zohar2016,Sala2018}:  
		They clarify that partial local charges or partial color cannot exist, yet entangled color singlets are allowed.	
	(2) \emph{High-energy physics}:  
		Color confinement \cite{Gross1973,Politzer1973}, black-hole pair production \cite{Hawking1975}, and hadronization can be viewed through the quantum-information lens of packaging.  
		Confinement, for instance, can be seen as the impossibility of producing a free single-particle excitation without forming a proper irrep carrying net color zero.
	(3) \emph{Quantum error-correction} \cite{NielsenChuang}:  
		Building gauge-invariant error-correcting codes requires that each logical excitation be consistently packaged.  
		This helps ensure that no unphysical states appear in the code space.

	\section{Conclusion}
	
	We have shown how local gauge invariance and superselection rules jointly yield the ``packaging principle'' for quantum field excitations:
	\emph{(1) No partial IQN factorization:}  
	A single-particle creation operator is an irreducible block under gauge transformations, forbidding partial factorization of electric charge, color, or baryon number.
	\emph{(2) Single net-charge superselection sector:}  
	Multi-particle states must lie entirely within one net charge (or color) sector, prohibiting cross-sector superpositions.
	\emph{(3) Packaged Entangled States:}  
	Within a single charge sector, non-factorizable superpositions remain possible, yielding packaged entangled excitations (e.g., mesons, electron-positron pairs, flavor-entangled $K^0 \bar{K}^0$, etc.).
	\emph{(4) Hybrid External-Internal DOFs:}  
	Spin or momentum can be combined with internal charges in each single-particle operator, leading to hybridized entanglement that remains gauge invariant (or transforms covariantly within a fixed charge sector) as long as the net charge is fixed.  
	Measuring an external DOF (e.g., spin) can collapse the entire external-internal wavefunction while preserving the superselection sector.

\end{document}